\documentclass{article}
\usepackage{fullpage}
\usepackage{url}
\usepackage{xspace}
\usepackage{color}

\usepackage{graphics}
\usepackage[dvips]{epsfig}

\usepackage{amsmath}
\usepackage{amssymb}
\usepackage{amsfonts}
\usepackage{graphicx}
\usepackage{algorithm}
\usepackage{comment}
\usepackage{multirow}
\usepackage{algpseudocode}

\newtheorem{theorem}{Theorem}[section]
\newtheorem{lemma}{Lemma}[section]
\newtheorem{corollary}{Corollary}[section]

\newtheorem{fact}{Fact}[section]

\newcommand{\qed}{\hfill $\Box$ \bigbreak}
\newenvironment{proof}{\noindent {\bf Proof.}}{\qed}

\newcommand{\remove}[1]{}

\newcommand{\const}{\ensuremath{\gamma}}
\newcommand{\probfrac}{\ensuremath{\frac{\epsilon}{4}}}

\title{Global Synchronization and Consensus Using Beeps in a Fault-Prone MAC}
\author{Kokouvi Hounkanli, Avery Miller, Andrzej Pelc\footnote{Partially supported by NSERC discovery grant 8136 -- 2013 and by the Research Chair in Distributed Computing at the Universit\'e du Qu\'{e}bec en Outaouais.}
\\Universit\'{e} du Qu\'{e}bec en Outaouais, Gatineau, Canada.}

\begin{document}

\maketitle

\begin{abstract}
Consensus is one of the fundamental tasks studied in distributed computing. Processors have input values from some set $V$ and they have to decide
the same value from this set. If all processors have the same input value, then they must all decide this value.
We study the task of consensus in a Multiple Access Channel (MAC) prone to faults, under a very weak communication model called the {\em beeping model}.
Communication proceeds in synchronous rounds. Some processors wake up spontaneously, in possibly different rounds decided by an adversary. In each round, an awake processor can either listen, i.e., stay silent,
or beep, i.e., emit a signal. In each round, a fault can occur in the channel independently with constant probability $0<p<1$. In a fault-free round,
an awake processor hears a beep if it listens in this round and if one or more other processors beep in this round.
A processor still dormant in a fault-free round in which some other processor beeps is woken up by this beep and hears it. In a faulty round nothing is heard, regardless of the behaviour of the processors.

An algorithm working with error probability at most $\epsilon$, for a given $\epsilon>0$, is called $\epsilon$-{\em safe}. Our main result is the design and analysis,
for any constant $\epsilon>0$, of
a {deterministic} $\epsilon$-safe consensus algorithm that works in time $O(\log w)$ in a fault-prone MAC, where $w$ is the smallest input value of all participating processors.
We show that this time cannot be improved, even when the MAC is fault-free. The main algorithmic tool that we develop to achieve our goal, and that might be of independent interest,
is a  deterministic algorithm that, with arbitrarily small constant error probability, establishes a global clock in a fault-prone MAC in constant time.

\vspace{2ex}

\end{abstract}

%
%
%


\section{Introduction}

\noindent
{\bf Background.}
Consensus is one of the fundamental tasks studied in distributed computing \cite{Ly}. 
Processors have input values from some set $V$, and they have to decide
the same value from this set. If all processors have the same input value, then they must all decide this value.
Consensus has mostly been studied in the context of fault-tolerance. Either the communication between processors is assumed prone to faults \cite{G,SW1,SW2},
or processors themselves can be subject to crash \cite{Cho08, MR} or Byzantine \cite{PSL} faults. In the present paper, we study a scenario falling under the first 
of these variants.

\noindent
{\bf Model and Problem Description.}
We study the task of consensus defined as follows \cite{Ly}. Processors have input values from some set $V$ of non-negative integers. The goal for all processors is to satisfy the following requirements.

\begin{itemize}
\item {\em Termination}: all processors must output some value from $V$.
\item {\em Agreement}: all output values must be equal.
\item {\em Validity}: if all input values are equal to $v$, then all output values must be equal to $v$.
\footnote{Some authors use a stronger validity condition in which the output values must always be one of the input values, even if these are non-equal. In this paper we use the above formulation
from \cite{Ly}. }
\end{itemize}

We study the task of consensus in a Multiple Access Channel (MAC). In a MAC, all processors can communicate directly, i.e., the underlying communication graph is complete. 
Communication proceeds in synchronous rounds. 
Some processors wake up spontaneously, in possibly different rounds decided by an adversary. Each processor has a local clock that starts at its wake-up, showing round number 0. All clocks tick at the same rate, one tick per round. This is a weak version of
synchrony, which should be contrasted with the assumption of a {\em global clock} where the clock of each processor shows a global round number that is equal for all of them.

We adopt a very weak communication model called the {\em beeping model}. We assume that processors are fault-free, while the MAC is prone to random faults. Faults in the channel may be due to some random noise occurring in the background.
In each round, an awake processor can either {\em listen}, i.e., stay silent,
or {\em beep}, i.e., emit a signal. In each round, a fault can occur in the channel independently with constant probability $0<p<1$. The value of $p$ is known by all processors. In a fault-free round,
an awake processor hears a beep if it listens in this round and if one or more other processors beep in this round. In a faulty round, nothing is heard regardless of the behaviour of the processors.
A processor that is still dormant in a fault-free round in which some other processor beeps is woken up by this beep and hears it.

The beeping model was introduced
in \cite{CK} for vertex coloring,  used
in \cite{AABCHK} to solve the MIS problem, and later used in \cite{FSW,GH} to solve leader election. The beeping model is widely applicable, as it makes small demands on communicating devices by relying only on carrier sensing. In fact, as mentioned in 
\cite{CK}, beeps are an even weaker way of communicating than using one-bit messages: one-bit messages allow three different states (0,1 and no message),
while beeps permit to differentiate only between a signal and its absence. 

We do not assume that processors in the channel have access to any random generator. We study deterministic consensus algorithms working in a probabilistic
fault-prone MAC, which work
with error probability at most $\epsilon$, for a given $\epsilon>0$. Such algorithms are called $\epsilon$-{\em safe}.  We assume that all processors know the value of $\epsilon$.

\noindent
\subsection{Our results} 
Our main result is the design and analysis,
for any constant $\epsilon>0$, of
a {deterministic} $\epsilon$-safe consensus algorithm that works in time $O(\log w)$ for a fault-prone MAC, where $w$ is the smallest input value of all participating processors.
We show that this time cannot be improved by a deterministic algorithm, even when the MAC is fault-free. Moreover, we show how to reach consensus in the same round. Hence, as formulated in \cite{MR},
we reach ``double agreement, one on the decided value (data agreement) and one on the decision round (time
agreement)''.

The main algorithmic tool that we develop to achieve our goal, and that might be of independent interest,
is a {deterministic} algorithm that, with arbitrarily small constant error probability, establishes a global clock in a fault-prone MAC in constant time.

\noindent
\subsection{Related work}
The Multiple Access Channel (MAC) is a popular and well-studied medium of communication. Most research concerning the MAC
has been done under the radio communication model in which processors can send an entire message in a single round, and this message
is heard by other processors if exactly one processor transmits, and all others listen in this round. This communication model is incomparable to the beeping model:
on the one hand it is much stronger, as large messages (and not only beeps) can be sent in a single round, but on the other hand it is weaker, as it requires
a unique transmitter in a round to make the transmission successful, while in the beeping model many beeps may be heard simultaneously.
Leader election was studied in a MAC under the radio model, both in the deterministic \cite{CMS,GW} and in the randomized setting \cite{BGI,Wil}.

Consensus is a classic problem in distributed computing, mostly studied assuming that processors
communicate by shared variables or through message passing networks \cite{AW,Ly}.
See the recent book \cite{Ra} for a comprehensive survey of the literature on consensus, mostly concerning processor faults. 
In \cite{GK}, the authors showed a randomized consensus for crash faults with optimal communication complexity.
In~\cite{Cho08}, the feasibility and complexity of consensus in a multiple access channel (MAC) with 
simultaneous wake-up and crash failures were studied in the context of different collision detectors.
Consensus and mutual exclusion in a MAC (without faults) were studied in \cite{CGKP}. The authors also investigated the impact of a global clock and of the capability of
collision detection on the time efficiency of consensus.
Consensus in the quantum setting has been studied, e.g., in \cite{CKS}. To the best of our knowledge, consensus with faulty beeps has never been studied before.

The differences between local and global clocks for the wake-up problem were first studied in \cite{GPP} and then in
\cite{CGKR,CK,CR}. The communication model used in these papers was that of radio networks in which the main challenge is the occurrence of collisions between
simultaneously received messages. A global clock is often used in the study of  broadcasting in radio networks
(cf. \cite{CR}).

\section{Global Synchronization}

In this section, we provide an algorithm \texttt{GlobalSync} that establishes a global clock. 
Upon its wake-up, each processor in the channel executes \texttt{GlobalSync} with its local clock initialized to 0.
The round in which the first wake-up occurs is defined as global round 0. Processors are not aware of the relationship between their local clock values and this 
global round. Establishing a global clock means that all processors in the channel exit \texttt{GlobalSync} in the same global round.

Fix any constant $\epsilon > 0$. Let $\const$ be a constant such that $p^{\const}< \probfrac$. Hence, in a sequence of $\const$ consecutive rounds of beeps, at least one of these beeps occurs in a fault-free round with probability
at least $1-\probfrac$.

We describe Algorithm \texttt{GlobalSync} whose aim is to ensure that all processors agree on a common global round, i.e. they establish a global clock.
At a high level, the algorithm proceeds as follows. 
A processor that wakes up spontaneously beeps periodically trying to wake up all other processors that are still dormant. These beeps will be called {\em alarm beeps}.
They are separated by time intervals of increasing size,  which prevents an adversary from setting wake-up times so that all alarm beeps are aligned. 
In the intervals between alarm beeps, the processor is waiting for a response from 
other processors to indicate that they heard an alarm beep. If a large enough number of such intervals occur without any response, then the processor assumes that 
the entire channel was woken up at the same time, and a global round is chosen as the round in which the next alarm beep is scheduled.
Otherwise, if a beep was heard in one of these intervals, the processor listens
for $2\const$ consecutive rounds and then beeps for $2\const$ consecutive rounds. Similarly, a processor woken up by a beep listens
for $2\const$ consecutive rounds and then beeps for $2\const$ consecutive rounds.
The global round chosen by the algorithm is the round $r+4\const +1$, where
$r$ is the first round when an alarm beep was heard by some processor. The difficulty is for each processor to determine the round $r$.
This is because, when a beep is heard, 
 there are two possible cases:
such a beep may be an alarm beep from another processor, or may be in response to an alarm beep. 
We overcome this difficulty as follows. Time is divided into blocks of $2\const$ consecutive rounds. If a single beep is heard in a block, the processor concludes that it was an alarm beep; if more than one beep is heard in a block, the processor concludes that these beeps were in response to an alarm beep. We will prove that
such conclusions are correct with sufficiently high probability.
Finally, each processor considers the first round $s$ in which it heard a beep. If this beep was an alarm beep, 
the processor sets $r=s$. If this beep was in response to an alarm beep, the processor sets $r$ to be the most recent round before $s$ in which it beeped. 

We now provide the details of Algorithm \texttt{GlobalSync}. The following procedure provides an aggregate count of the beeps recently heard by a processor. More specifically, for a given round $t'$, the next $4\const$ rounds
are treated as two blocks of $2\const$ rounds each, and for each block, the cases of 0, 1, or more beeps are distinguished.

\begin{algorithm}[H]
\caption{\texttt{listenVector}$(t')$}
\begin{algorithmic}[1]
\State $h_1 \leftarrow 0$
\State $h_2 \leftarrow 0$
\State $num_1 \leftarrow $ number of beeps heard in rounds $t',t'+1,\ldots,t'+2\const-1$
\State $num_2 \leftarrow $ number of beeps heard in rounds $t'+2\const,\ldots,t'+4\const-1$
\State {\bf if} $num_1 = 1$, then $h_1 \leftarrow 1$
\State {\bf if} $num_1 > 1$, then $h_1 \leftarrow *$
\State {\bf if} $num_2 = 1$, then $h_2 \leftarrow 1$
\State {\bf if} $num_2 > 1$, then $h_2 \leftarrow *$
\State {\bf return} $[h_1\ h_2]$
\end{algorithmic}
\label{listenVector}
\end{algorithm}


Below we give the pseudocode of Algorithm \texttt{GlobalSync} using the above procedure.

\begin{algorithm}[H]
\caption{\texttt{GlobalSync}}
\begin{algorithmic}[1]

\State {\bf if} woken up by a beep in some round $heard$: \Comment{woken up by beep} \label{wokenbybeep}

\State \indent beep $2\const$ consecutive rounds starting at round $heard+2\const+1$ \label{wokenresponse}
\State \indent $syncRound \leftarrow heard+4\const+1$\label{wokenoutput}

\State{\bf else}: \Comment{woken up spontaneously}
\State \indent $i \leftarrow 0$
\State \indent $myNextBeep \leftarrow 0$
\State \indent {\bf repeat}: \label{loop}
\State \indent \indent $myCurrentBeep \leftarrow myNextBeep$
\State \indent \indent beep in round $myCurrentBeep$ \label{alarm}
\State \indent \indent $i \leftarrow i+1$
\State \indent \indent $myNextBeep \leftarrow myCurrentBeep + 4\const + i$
\State \indent  {\bf until} ($i=3\const$) or (a beep is heard in one of $\{myCurrentBeep+1,\ldots,myNextBeep-1\})$
\State \indent {\bf if} $i=3\const$: \label{exhausted}
\State \indent \indent $syncRound \leftarrow myNextBeep$ \label{exhaustedoutput}
\State \indent {\bf else}: \label{elseline}
\State \indent \indent $heard \leftarrow$ first round after $myCurrentBeep$ in which a beep was heard
\State \indent \indent $[h_1\ h_2] = listenVector(myCurrentBeep+1)$
\State \indent \indent {\bf if} $[h_1\ h_2] \in \{[0\ 0], [0\ 1], [1\ 0], [1\ 1], [1\ *]\}$: \label{IfNotAlarm}
\State \indent \indent \indent beep $2\const$ consecutive rounds starting at round $heard+2\const+1$\label{heardresponse}
\State \indent \indent \indent $syncRound \leftarrow heard+4\const+1$ \label{heardoutput}
\State \indent \indent {\bf if} $[h_1\ h_2] \in \{[0\ *], [*\ 0], [*\ 1], [*\ *]\}$: \label{IfAlarm}
\State \indent \indent \indent $syncRound \leftarrow myCurrentBeep+4\const+1$ \label{alarmoutput}

\State wait until round $syncRound$ and {\bf exit} 

\end{algorithmic}
\label{synchronize}
\end{algorithm}

In the analysis of Algorithm  \texttt{GlobalSync} we refer to global rounds, but it should be recalled that processors in the channel  
do not have access to the global clock values: all a processor sees is its local clock.  
The following fact follows from the algorithm description by induction on $i$.

\begin{fact}\label{myNextBeep}
At the end of each loop iteration, the variable $myNextBeep$ is equal to $4\const i + \sum_{k=1}^{i} k$. Further, if a processor is woken up at time $t$, then, at the end of each loop iteration, $myNextBeep$ is equal to the local clock value corresponding to the global round $t + 4\const i + \sum_{k=1}^{i} k$.
\end{fact}

We say that a processor is \emph{lonely in round $t$} if it has not heard a beep in any round up to and including round $t$. Using Fact \ref{myNextBeep}, we can determine the number of rounds that elapse before a lonely processor beeps a given number of times.

\begin{fact}\label{beepbound}
Suppose that a processor $v$ wakes up spontaneously in round $t_1$. If $v$ is lonely in round $t_1 + 4\const i + i(i+1)/2$, then $v$ has beeped exactly $i$ times before this round.
\end{fact}

The next lemma shows that, for a certain time interval after the first wake-up, no processor terminates its execution of \texttt{GlobalSync} without first hearing a beep.  

\begin{lemma}\label{notermination}
Suppose that the first spontaneous wake-up occurs in round $t_1$.
Then no processor terminates its execution of \texttt{GlobalSync} in the time interval $[t_1,\dots, t_1+ 12\const ^2 + (3\const)(3\const+1)/2 -1]$ without hearing a beep.
\end{lemma}
\begin{proof}
If a processor $v$ terminates its execution of \texttt{GlobalSync} without hearing a beep, then its {\bf repeat} loop exited with $i=3\const$.
By line \ref{exhaustedoutput}, the processor will terminate in round $myNextBeep$, which, by Fact \ref{myNextBeep}, corresponds to the global round
 $t_v + 12\const ^2 + (3\const)(3\const+1)/2$, where $t_v \geq t_1$ is the wake-up round of processor $v$.
\end{proof}

In order to prove the correctness of the algorithm, we first consider the case when all processors wake up spontaneously in the same round.

\begin{lemma}\label{simultaneous}
Suppose that all processors wake up spontaneously in the same global round $t_1$. With probability 1, all processors terminate their execution of \texttt{GlobalSync} in global round $t_1 + 12\const ^2 + (3\const)(3\const+1)/2$.
\end{lemma}
\begin{proof}
Since every processor is woken up spontaneously in global round $t_1$, the {\bf if} condition on line \ref{wokenbybeep} evaluates to false at every processor. Therefore, all processors execute the loop at line \ref{loop}. In particular, this means that all processors beep in their local round 0. By Fact \ref{myNextBeep}, at the end of each loop iteration, the variable $myNextBeep$ at every processor is equal to the local clock value corresponding to the global round $t_1 +4\const i + \sum_{k=1}^{i} k$. It follows that all processors beep in the same rounds. In particular, this means that no processor ever hears a beep. Thus, at every processor, the loop exits with $i=3\const$. So, the {\bf if} condition on line \ref{exhausted} evaluates to true. By line \ref{exhaustedoutput}, each processor sets $syncRound$ to the value $4\const (3\const) + \sum_{k=1}^{3\const} k = 12\const ^2 + (3\const)(3\const+1)/2$, which is their local clock value that corresponds to the global round $t_1 + 12\const ^2 + (3\const)(3\const+1)/2$. Therefore, all processors terminate their execution of \texttt{GlobalSync} in global round $t_1 + 12\const ^2 + (3\const)(3\const+1)/2$.
\end{proof}

Note that, when our algorithm is executed in the case where all processors wake up spontaneously in the same round, no processor ever hears a beep, and, after a fixed length of silence, all processors terminate their execution of \texttt{GlobalSync}. In the case where not all processors wake up spontaneously in the same round, if the same fixed length of silence is observed by all processors, then, again, all processors will terminate their execution of \texttt{GlobalSync}, but this time in different rounds. This would be a bad case for our algorithm. We now show that, with sufficiently high probability, such a bad case does not occur, i.e., that there exists some round $t^*$ in which a beep is heard by some processor.

\begin{lemma}\label{goodround}
Suppose that not all processors wake up spontaneously in the same round, and suppose that the first spontaneous wake-up occurs in some round $t_1$. With probability at least $(1-\frac{\epsilon}{2})$, there exists a global round $t^* \leq t_1 + 12\const^2 + (3\const)(3\const+1)/2$ in which all of the following hold: no processor has terminated its execution of \texttt{GlobalSync}, at least one processor beeps, at least one processor listens, and no fault occurs.
\end{lemma}
\begin{proof}
By Lemma \ref{notermination}, if at least one processor terminates its execution of \texttt{GlobalSync} in the interval $[t_1,\ldots,t_1+12\const^2 + (3\const)(3\const+1)/-1]$, then there exists a round $t^*$ in this interval before the first such termination with the property that at least one processor beeps, at least one processor listens, and no fault occurs, as claimed.

So, we proceed with the assumption that no processor terminates its execution of \texttt{GlobalSync} in the interval $[t_1,\ldots,t_1+12\const^2 + (3\const)(3\const+1)/-1]$. Let $t_2$ be the first round after $t_1$ such that some processor wakes up in round $t_2$. Let $v_1$ be a processor that wakes up in round $t_1$, and let $v_2$ be a processor that wakes up in round $t_2$. 

If processor $v_1$ hears a beep in some round $t^*$ in the interval $[t_1,\ldots,t_1+4\const^2+(\const)(\const+1)/2]$, we are done. So, in the rest of the proof, we assume that $v_1$ is lonely in round $t_1+4\const^2+(\const)(\const+1)/2$. By Fact \ref{beepbound}, $v_1$ beeps exactly $\const$ times in the interval $[t_1,\ldots,t_1+4\const^2+(\const)(\const+1)/2-1]$.

If processor $v_2$ hears a beep in some round $t^*$ in the interval $[t_1,\ldots,t_1+12\const^2+(3\const)(3\const+1)/2-1]$, we are done. Hence, in the rest of the proof, we assume that $v_2$ is lonely in round $t_1+12\const^2+(3\const)(3\const+1)/2-1$. This assumption implies that, if $t_2 \leq t_1+12\const^2+(3\const)(3\const+1)/2-1$, then $v_2$'s wake-up is spontaneous.


First, we show that, with probability at least $1-\probfrac$, we have $t_2 < t_1 + 4\const^2 + (\const)(\const+1)/2$. To see why, recall that $v_1$ beeps $\const$ times in the interval $[t_1,\ldots,t_1+4\const^2+(\const)(\const+1)/2-1]$. With probability at least $1-\probfrac$, one of these first $\const$ beeps by $v_1$ is in a fault-free round. Whenever one of the first $\const$ beeps by $v_1$ is in a fault-free round $s$, we have $t_2 < s$, since, otherwise, $v_2$ would be woken up by a beep in round $s \leq t_1 + 12\const^2 + (3\const)(3\const+1)/2 - 1$, which contradicts our above assumption about $v_2$'s wake-up.

Next, we note that, with probability at least $1-\probfrac$, the first $2\const$ rounds in which $v_2$ beeps occur before round $t_1 + 12\const^2 + (3\const)(3\const+1)/2$ (i.e., before any processor has terminated its execution of \texttt{GlobalSync}). This is because we have already shown that, with probability at least $1-\probfrac$, we have $t_2 < t_1 + 4\const^2 + (\const)(\const+1)/2$, and, by Fact \ref{beepbound}, the first $2\const$ beeps by $v_2$ occur by round $t_2 + 8\const^2 + (2\const)(2\const+1)/2$.

Next, we show that, with probability at least $1-\probfrac$, in one of the first $2\const$ rounds in which $v_2$ beeps, processor $v_1$ listens and no fault occurs. We first note that, in at least $\const$ of the first $2\const$ rounds in which $v_2$ beeps, processor $v_1$ listens. To see why, we show that no two consecutive beeps by $v_2$ can occur in the same rounds as beeps by $v_1$. If $v_2$ beeps in some global round $t_2 + 4\const i_2 + \sum_{k=1}^{i_2} k$, and this is equal to some global round $t_1 + 4\const i_1 + \sum_{k=1}^{i_1} k$ in which $v_1$ beeps, then, since $t_2 > t_1$, we must have $i_1 > i_2$. It follows that $v_2$'s next beep will occur in round $[t_2 + 4\const (i_2) + \sum_{k=1}^{i_2} k] + 4\const + (i_2+1) =  [t_1 + 4\const (i_1) + \sum_{k=1}^{i_1} k] + 4\const + (i_2+1) < [t_1 + 4\const (i_1) + \sum_{k=1}^{i_1} k] + 4\const + (i_1+1)$, i.e., before $v_1$'s next beep. This proves that, in at least $\const$ of the first $2\const$ rounds in which $v_2$ beeps, processor $v_1$ listens. With probability at least $(1-\probfrac)$, at least one of these $\const$ beeps occurs in a fault-free round. Thus, with probability at least $(1-\probfrac)$, in one of the first $2\const$ rounds in which $v_2$ beeps, processor $v_1$ listens and no fault occurs.

Altogether, we have shown that, with probability at least $1-\probfrac$, the first $2\const$ rounds in which $v_2$ beeps occur before round $t_1 + 12\const^2 + (3\const)(3\const+1)/2$, and that, with probability at least $1-\probfrac$, one of the first $2\const$ beeps by $v_2$ is heard by $v_1$. It follows that, with probability at least $(1-\frac{\epsilon}{2})$, there is a round $t^* \leq t_1 + 12\const^2 + (3\const)(3\const+1)/2$ in which $v_2$ beeps, $v_1$ listens, no fault occurs, and no processor has terminated its execution of \texttt{GlobalSync}.
\end{proof}

We now proceed to prove the correctness of our algorithm for the case where not all processors wake up spontaneously in the same round. We will be able to do so when there exists a global round $t^*$ satisfying the conditions specified in Lemma \ref{goodround}. We will use the following lemma that establishes listening periods of processors.

\begin{lemma}\label{silence}
Suppose that not all processors wake up spontaneously in the same round. Let $t^*$ be the first global round in which all of the following hold: no processor has terminated its execution of \texttt{GlobalSync}, at least one processor beeps, at least one processor listens, and no fault occurs. Then, no processor beeps in rounds $t^*+1,\ldots,t^*+2\const$.
\end{lemma}
\begin{proof}
There are several cases to consider. First, we consider each processor $v$ that does not beep in round $t^*$, and show that $v$ does not beep before round $t^*+2\const+1$. Suppose that $v$ is woken up by a beep in round $t^*$. In this case, by line \ref{wokenresponse}, $v$ waits $2\const$ rounds before its next beep, as claimed. Next, suppose that $v$ hears a beep in round $t^*$, and that $v$ was woken up before round $t^*$.
This case corresponds to the \textbf{else} clause at line \ref{elseline}. Note that the local clock value corresponding to $t^*$ is stored in $v$'s $heard$ variable. Since $v$'s next beep occurs at line \ref{heardresponse},  it follows that $v$ waits $2\const$ rounds after round $t^*$ before its next beep, as claimed.

Finally, consider the case where $v$ beeped in round $t^*$. This occurs at line \ref{alarm}, and note that the local clock value corresponding to $t^*$ is stored in $v$'s $myCurrentBeep$ variable.  If $v$ does not hear a beep between rounds $myCurrentBeep$ and $myNextBeep$, then $v$ will not beep again until round $myNextBeep = myCurrentBeep + 4\const + i > t^* + 2\const$, as claimed. If $v$ does hear a beep between rounds $myCurrentBeep$ and $myNextBeep$, say in some round $heard$, then $v$'s next beep occurs at line \ref{heardresponse}. This beep is in round $heard + 2\const + 1 > myCurrentBeep + 2\const + 1 = t^* + 2\const + 1$, as claimed.
\end{proof}

The next lemma shows that all processors terminate their execution of \texttt{GlobalSync} in the same global round soon after $t^*$.

\begin{lemma}\label{sameround}
Suppose that not all processors wake up spontaneously in the same round. Let $t^*$ be the first global round in which all of the following hold: no processor has terminated its execution of \texttt{GlobalSync}, at least one processor beeps, at least one processor listens, and no fault occurs. With probability at least $(1-\frac{\epsilon}{2})$, all processors terminate their execution of \texttt{GlobalSync} in global round $t^*+4\const+1$.
\end{lemma}
\begin{proof}
First, consider any processor $v$ that does not beep in global round $t^*$. We will show that $v$ beeps in rounds $t^*+2\const+1,\ldots,t^*+4\const$, and that $v$ sets $syncRound$ to the local clock value that corresponds to the global round $t^*+4\const+1$. There are two cases to consider:
\begin{enumerate}
\item {\bf Suppose that $v$ wakes up in global round $t^*$.} This occurs at line \ref{wokenbybeep}. Note that, in round $t^*$, the local clock value corresponding to $t^*$ is stored in $v$'s $heard$ variable. Then, by line \ref{wokenresponse}, processor $v$ beeps in rounds $t^*+2\const+1,\ldots,t^*+4\const$, and, by line \ref{wokenoutput}, $v$ sets $syncRound$ to $heard + 4\const + 1$, which is the local clock value that corresponds to the global round $t^*+4\const+1$, as claimed.

\item {\bf Suppose that $v$ was woken before global round $t^*$.} Let $s$ be the latest round before $t^*$ during which $v$ beeped. Note that, in round $t^*$, the local clock value corresponding to $s$ is stored in $v$'s $myCurrentBeep$ variable. By the choice of $s$ and $t^*$, the beep heard by $v$ during $t^*$ is the first beep that $v$ hears between rounds $myCurrentBeep$ and $myNextBeep$. Thus, the local clock value corresponding to $t^*$ is stored in $v$'s $heard$ variable. By Lemma \ref{silence}, no beeps occur in the $2\const$ rounds following round $heard$. In particular, this means that the first non-zero entry of $[h_1\ h_2]$ cannot be a $*$, so the {\bf if} condition on line \ref{IfAlarm} evaluates to false. Since the {\bf if} conditions on lines \ref{IfNotAlarm} and \ref{IfAlarm} exhaust all 9 possibilities for the vector $[h_1\ h_2]$, it follows that the {\bf if} condition on line \ref{IfNotAlarm} evaluates to true. Therefore, by line \ref{heardresponse}, $v$ beeps in rounds $t^* + 2\const + 1,\ldots, t^*+4\const$, and, by line \ref{heardoutput}, $v$ sets $syncRound$ to $heard+4\const+1$, which is the local clock value that corresponds to the global round $t^*+4\const+1$, as claimed.
\end{enumerate}

By the choice of $t^*$, there is at least one processor that does not beep in global round $t^*$. From what we have just shown, it follows that at least one processor $v$ beeps in rounds $t^*+2\const+1,\ldots,t^*+4\const$. Next, we show that, with probability at least $(1-\frac{\epsilon}{2})$, two or more of these beeps occur in fault-free rounds. Consider the intervals $[t^*+2\const+1,\ldots,t^*+3\const]$ and $[t^*+3\const+1,\ldots,t^*+4\const]$. Processor $v$ beeps $\const$ times in each of these intervals. It follows that, for each of these intervals, $v$ beeps in a fault-free round with probability at least $1-\probfrac$. Therefore, with probability at least $(1-\frac{\epsilon}{2})$, in the interval $[t^*+2\const+1,\ldots,t^*+4\const]$, two or more beeps by $v$ occur in fault-free rounds.

Finally, we show that if two or more of the beeps in rounds $t^*+2\const+1,\ldots,t^*+4\const$ occur in fault-free rounds, then every processor $v$ that beeps in round $t^*$ sets $syncRound$ to the local clock value that corresponds to the global round $t^* + 4\const + 1$. Consider any processor $v$ that beeps in round $t^*$, and suppose that two or more of the beeps in rounds $t^*+2\const+1,\ldots,t^*+4\const$ are successful. The beep by $v$ in round $t^*$ occurs at line \ref{alarm}, and, note that in round $t^*$,  the local clock value corresponding to $t^*$ is stored in $v$'s $myCurrentBeep$ variable. By Lemma \ref{silence}, no beeps occur in rounds $t^*+1,\ldots,t^*+2\const$. Since two or more beeps in rounds $t^*+2\const+1,\ldots,t^*+4\const$ occur in fault-free rounds, it follows that, at processor $v$, the vector $[h_1\ h_2]$ is equal to $[0\ *]$. Therefore, the {\bf if} condition on line \ref{IfAlarm} evaluates to true. Therefore, by line \ref{alarmoutput}, processor $v$ sets $syncRound$ to $myCurrentBeep+4\const+1$, which is the local clock value that corresponds to the global round $t^*+4\const+1$, as claimed.
\end{proof}

Finally, we show that Algorithm \texttt{GlobalSync} runs in constant time and fails with probability at most $\epsilon$ for any given constant $\epsilon > 0$.

\begin{theorem}\label{sync}
Fix any constant $\epsilon > 0$. With probability at least $1-\epsilon$, all processors terminate Algorithm \texttt{GlobalSync} in the same global round {\tt sync},
which occurs  $O(1)$ rounds after the first wake-up.
\end{theorem}
\begin{proof}
Let $t_1$ be the first round in which a wake-up occurs. In the case where all processors wake up spontaneously in the same round, Lemma \ref{simultaneous} implies that, with probability 1, all processors terminate Algorithm \texttt{GlobalSync} in global round {\tt sync} $=t_1+12\const^2+(3\const)(3\const+1)/2$. In the case where not all processors wake up spontaneously in the same round, Lemmas \ref{goodround} and \ref{sameround} imply that all processors terminate Algorithm \texttt{GlobalSync} in global round  {\tt sync} $=t^*+4\const+1$, where $t^*\leq t_1+12\const^2+(3\const)(3\const+1)/2$, with error probability at most $\frac{\epsilon}{2}+\frac{\epsilon}{2}=\epsilon$.
\end{proof}

\section{Consensus}

In this section, we provide a deterministic decision procedure which achieves consensus assuming that global synchronization has been done previously. It is performed after Algorithm {\tt GlobalSync} and has the following property.
Let ${\tt sync}$ be the global round in which all processors in the channel terminate their execution of Algorithm {\tt GlobalSync}. Algorithm {\tt Decision} achieves consensus with error probability at most $\epsilon$ in the global
round $s={\tt sync}+O(\log w)$, where $w$ is the smallest of all input values of processors in the channel.

Consider the input value $val$ of a processor $v$ and let  $\mu=(a_1,\dots , a_m)$ be its binary representation. We transform the sequence $\mu$ by replacing each bit 1 by $(10)$, each bit 0 by $(01)$ and appending
$(11)$ at the end. Hence the transformed sequence is
$(c_1,\dots ,c_{2m+2})$, where\\
$c_i=1$, for $i \in \{2m+1,2m+2\}$, and,\\
for $j=1,\dots , m$:\\
$c_{2j-1}=1$ and $c_{2j}=0$, if $a_j=1$,\\
$c_{2j-1}=0$ and $c_{2j}=1$, if $a_j=0$.

The sequence $(c_1,\dots ,c_{2m+2})$ is called the {\em transformed input value} of processor $v$ and is denoted by $val^*$.
Notice that if the input values of two processors are different, then there exists an index for which the corresponding bits of their transformed input values differ (this is not necessarily the case for the original input values, since one of the binary representations
might be a prefix of the other).

The high-level idea of Algorithm {\tt Decision} is the following. 
A processor beeps and listens in time intervals of prescribed length, starting in global round ${\tt sync}+1$, according to its transformed input value. If it does not hear any beep, it concludes that all input values are identical and outputs its input value. Otherwise, it concludes that there are different input values and then 
outputs a default value. We will prove that these conclusions are correct with probability at least $1-\epsilon$, and that all processors make the decision in a common global round $s={\tt sync}+O(\log w)$.

We now give the pseudocode of the algorithm executed by a processor whose input value is $val$. We assume that the algorithm is started in global round ${\tt sync}+1$, and we let $r$ be the processor's local clock value corresponding to the global round ${\tt sync}$. Let $x$ be the smallest positive integer such that $p^x <\epsilon /2$.
Let $val_0$ be the smallest integer in $V$, which we will use as the default decision value.

\begin{algorithm}[H]
\caption{\texttt{Decision}}
\begin{algorithmic}[1]

\State $(c_1,\dots ,c_k) \leftarrow val^*$
\State $i \leftarrow 1$
\State $heard \leftarrow \mathit{false}$
\State {\bf while} ($heard=$ {\em false} {\bf and} $i\leq k$) {\bf do}
\State \indent {\bf if} $c_i=1$ {\bf then} beep for $x$ rounds and then listen for $x$ rounds
\State \indent {\bf if} $c_i=0$ {\bf then} listen for $x$ rounds and then beep for $x$ rounds
\State \indent {\bf if} a beep was heard {\bf then} $heard \leftarrow$ {\em true}
\State \indent $i \leftarrow i+1$
\State {\bf if} $heard=$ {\em false} {\bf then} output $val$ in round $r+2(i-1)x+1$ \label{didnthear}
\State {\bf else} output $val_0$ in round $r+2(i-1)x+1$ \label{didhear}

\end{algorithmic}
\label{decision}
\end{algorithm}

The following result shows that, with error probability at most $\epsilon$, upon completion of Algorithm {\tt Decision}, all processors in the channel correctly solve consensus in the same round, and this round occurs $O(\log w)$ rounds after global round ${\tt sync}$, where $w$ is the smallest of all input values of processors in the channel.

\begin{theorem}
Let ${\tt sync}$ be the common global round in which all processors terminate their execution of Algorithm {\tt GlobalSync}, and let $w$ be the smallest of all input values of processors in the channel. There exists a global round $s={\tt sync}+O(\log w)$ such that, with probability at least $1-\epsilon$, upon completion of Algorithm {\tt Decision}, all processors in the channel output the same value in global round $s$, and this value is their common input value if all input values were identical.
\end{theorem}

\begin{proof}
First, suppose that the input values of all processors in the channel are identical. Let $k \in O(\log w)$ be the length of their common transformed input value. 
Then each processor leaves the {\bf while} loop with the value of the variable $heard$ equal to false, and consequently, at line \ref{didnthear}, it outputs the common input value in round $r+2xk+1$, which is its local clock value corresponding to global round ${\tt sync}+2xk+1$.
Since $x$ is a constant, we have $2xk+1 \in O(\log w)$, which concludes the proof in this case.

In the remainder of the proof, we suppose that there are at least two distinct input values. Let $k_1$ be the length of the transformed input value $w^*$ corresponding to the input value $w$.  
Consider all transformed input values of processors in the channel, and let $j\leq k_1$ be the first index in which two of these transformed input values differ. 

For any $t>0$, let $A_t$ be the global time interval $\{{\tt sync}+2x(t-1)+1,\dots, {\tt sync}+2x(t-1)+x\}$, and let $B_t$ be the global time interval $\{{\tt sync}+2x(t-1)+x+1,\dots, {\tt sync}+2x(t-1)+2x\}$.
Let $E$ be the event that at least one round in the time interval $A_j$ is fault free and at least one round in the time interval $B_{j}$ is fault free.
By the definition of $x$, the probability of event $E$ is at least $1-\epsilon$. Suppose that event $E$ holds. By the choice of $j$, no beep was heard in the channel in global rounds $\{{\tt sync}+1,\ldots,{\tt sync}+2x(j-1)\}$, hence all processors in the channel participate in the $j^{th}$ iteration of the loop.
Consider any processor $v$ for which the $j^{th}$ bit of its transformed input value
is 0 and any processor $v'$ for which the $j^{th}$ bit of its transformed input value is 1. Processor $v$ listens in all rounds of $A_j$ and beeps in all rounds of $B_j$, whereas processor $v'$ beeps in all rounds of $A_j$ and listens
in all rounds of $B_j$. Hence, $v$ hears at least one beep in the time interval $A_j$, and $v'$ hears at least one beep in the time interval $B_j$. Therefore, both $v$ and $v'$ set $heard$ equal to true in iteration $j$ of the {\bf while} loop. Consequently, each processor outputs the default value $val_0$ at line \ref{didhear} in round $r+2xj+1$, which is its local clock value corresponding to global round ${\tt sync} + 2xj+1$. Since $x$ is constant and $j \leq k_1 \in O(\log w)$, we have 
$2xj+1 \in O(\log w)$, which concludes the proof in the case where there are at least two distinct input values.
\end{proof}

Finally, given a bound $\epsilon >0$ on error probability of consensus, we first run Algorithm  {\tt GlobalSync} and then Algorithm {\tt Decision}, each with error
probability bound $\frac{\epsilon}{2}$, to get the following corollary.

\begin{corollary}
Fix any constant $\epsilon >0$ and consider a fault-prone MAC with communication by beeps, where $w$ is the smallest of all input values of processors in the channel. 
With error probability at most $\epsilon$, consensus can be solved deterministically in the same round, $O(\log w)$ rounds after the first wakeup.
\end{corollary}

We conclude this section by showing that, even in a model where every round in the MAC is fault-free and all processors are woken up spontaneously in the same round, deterministic consensus with $m$-bit inputs requires $\Omega(m)$ rounds, which implies that our consensus algorithm has optimal time complexity.

\begin{theorem}
Consensus with $m$-bit inputs in a fault-free MAC with beeps requires $\Omega(m)$ rounds.
\end{theorem}
\begin{proof}
Consider any consensus algorithm $\mathcal{A}$. Assume that, for every $m$-bit input value $s$, the execution of $\mathcal{A}$ with input $s$ by a single processor on the channel uses $o(m)$ rounds. For any input $s$, let $\mathrm{Pattern}(s)$ be the beeping pattern of a processor that is alone on the channel and executes $\mathcal{A}$ with input $s$. By the Pigeonhole Principle, there exist distinct $m$-bit inputs $a$ and $b$ such that $\mathrm{Pattern}(a) = \mathrm{Pattern}(b)$. 

For each $s \in \{a,b\}$, let $\alpha_s$ be the execution of $\mathcal{A}$ in the case where a processor $v_s$ is alone on the channel and is given input $s$. By Validity, for each $s \in \{a,b\}$, at the end of execution $\alpha_s$, processor $v_s$ must output $s$. Next, consider the execution $\alpha_{a,b}$ of $\mathcal{A}$ in the case where processors $v_a$ and $v_b$ are on the channel and are given inputs $a$ and $b$, respectively. Since $\mathrm{Pattern}(a) = \mathrm{Pattern}(b)$, it follows that executions $\alpha_a$ and $\alpha_{a,b}$ are indistinguishable to processor $v_a$, and that executions $\alpha_b$ and $\alpha_{a,b}$ are indistinguishable to processor $v_b$. Therefore, in execution $\alpha_{a,b}$, processor $v_a$ outputs $a$ and processor $v_b$ outputs $b$, which contradicts Agreement. Therefore, we incorrectly assumed that, for every $m$-bit input $s$, the execution of $\mathcal{A}$ with input $s$ by a single processor on the channel uses $o(m)$ rounds. It follows that there exists an execution of $\mathcal{A}$ that uses $\Omega(m)$ rounds, as claimed.
\end{proof}


\bibliographystyle{plain}

\begin{thebibliography}{12}

\bibitem{AABCHK}
Y. Afek, N. Alon, Z. Bar-Joseph, A. Cornejo, B. Haeupler, F. Kuhn, Beeping a maximal independent set.
Proc. 25th International Symposium on Distributed Computing (DISC 2011), LNCS 6950, 32-50.


 \bibitem{AW}
H. Attiya and J. Welch,
Distributed Computing, 2004,
John Wiley and Sons, Inc.


\bibitem{BGI}
R. Bar-Yehuda, O. Goldreich, A. Itai,
On the time complexity of broadcast in radio networks:
an exponential gap between determinism and randomization,
Journal of Computer and System Sciences 45 (1992) 104 - 126.











%
%




 










 
 


\bibitem{CGKR}
B.S. Chlebus, L. Gasieniec, D.R. Kowalski and T. Radzik,
On the wake-up problem in radio networks,  
{\em Proc.  32nd Colloquium on Automata, Languages and Programming (ICALP 2005)}, LNCS 3580, 347 - 359. 

 
 \bibitem{CKS}
B.S. Chlebus,  D.R. Kowalski and M. Strojnowski, 
Scalable quantum consensus with crash failures,
{\em Proc. 24th International Symposium on Distributed Computing (DISC 2010)}, LNCS 6343, 236-250.

 
 \bibitem{Cho08}
G. Chockler, M. Demirbas, S. Gilbert, N.A. Lynch, C.C. Newport and T. Nolte, 
Consensus and collision detectors in radio networks, 
{\em Distributed Computing}, 21 (2008), 55 - 84.

\bibitem{CMS}
A.E.F. Clementi, A. Monti, R. Silvestri, 
Selective families, superimposed codes, and broadcasting on unknown
radio networks,
Proc. 12th Ann. ACM-SIAM Symp. on Discrete Algorithms (SODA 2001),
709 - 718.
 

\bibitem{CK}
A. Cornejo, F. Kuhn, Deploying wireless networks with beeps,
Proc. 24th International Symposium on Distributed Computing (DISC 2010), LNCS 6343, 148-162.



\bibitem{CR}
A. Czumaj and W. Rytter, Broadcasting algorithms in radio networks with un- 
known topology, 
{\em Proc. 44th IEEE Symposium on Foundations of Computer Science (FOCS 2003)}492 - 501.





\bibitem{CGKP}
J. Czyzowicz, L. Gasieniec, D. Kowalski, A. Pelc, Consensus and mutual exclusion in a multiple access channel, 
IEEE Transactions on Parallel and Distributed Systems 22 (2011), 1092-1104. 

\bibitem{FSW}
K.-T. Forster, J. Seidel, R. Wattenhofer,
Deterministic leader election in multi-hop beeping networks,
Proc.28th International Symposium on Distributed Computing (DISC 2014), LNCS 8784, 212-226.





\bibitem{GPP}
L. Gasieniec, A. Pelc and D. Peleg, The wakeup problem in synchronous broadcast systems, {\em SIAM Journal on Discrete Mathematics} 14 (2001), 207-222. 

\bibitem{GH}
M. Ghaffari and B. Haeupler,
Near optimal leader election in multi-hop radio networks,
Proc. 24th Annual ACM-SIAM Symposium on Discrete Algorithms (SODA 2013),  748-766.



\bibitem{GK}
S. Gilbert and D. R. Kowalski, Distributed agreement with optimal communication complexity,
{\em Proc. 21st ACM-SIAM Symposium on Discrete Algorithms (SODA 2010)}, 965-977.

\bibitem{G}
J.N. Gray, Notes on data base operating systems, In: R. Bayer,  R.M.Graham and G. Seegmuller, Eds.,
Operating systems: An Advance Course, LNCS 60, p. 465. 

\bibitem{GW}
A.G. Greenberg, S. Winograd,  
A lower bound on the time needed in the worst case to resolve conflicts 
deterministically in multiple access channels,
Journal of the ACM 32 (1985) 589 - 596. 



\bibitem{Ly}
N.A. Lynch, 
Distributed Algorithms,
Morgan Kaufmann Publ., Inc., 1996.

\bibitem{MR}
Y. Moses, M. Raynal, Revisiting simultaneous consensus with crash failures, Journal of Parallel and Distributed Computing 69 (2009), 400-409.

\bibitem{PSL}
M.C. Pease, R.E. Shostak, L. Lamport, 
Reaching agreement in the presence of faults, 
{\em J. ACM} 27 (1980), 228 - 234.

\bibitem{Ra}
M. Raynal,  Fault-Tolerant Agreement in Synchronous Distributed Systems, Morgan \& Claypool Publishers 2010.

\bibitem{SW1}
N. Santoro, P. Widmayer, Time is not a healer, Proc. 6th Annual Symposium on Theoretical Aspects of Computer Science (STACS 1989), 304-313.

\bibitem{SW2}
N. Santoro, P. Widmayer, Distributed function evaluation in presence of transmission faults, Proc.  International Symposium on Algorithms (SIGAL 1990), 358--367.

\bibitem{Wil}
D.E. Willard, 
Log-logarithmic Selection Resolution Protocols in a Multiple Access Channel,
SIAM J. on Computing 15 (1986), 468-477. 









\end{thebibliography}


\end{document}